\documentclass[]{aiaa-tc}
\usepackage{amssymb}

\usepackage{amsmath}
\usepackage{amsfonts}
\usepackage{amsthm}

\title{Distributed MIN-MAX Optimization Application to Time-optimal Consensus: An Alternating Projection Approach}

 \author{
  Hu Chunhe%
    \thanks{PhD student,  School of Automation Science and Electrical Engineering, HCHH@asee.buaa.edu.cn.}
  \ and
  Chen Zongji
    \thanks{Professor, School of Automation Science and Electrical Engineering, czj@buaa.edu.cn.}\\
  {\normalsize\itshape
   Beihang University (BUAA), Beijing, 100191,
   China}\\
   {\normalsize\itshape
   National Key Laboratory of Science and Technology on Aircraft Control, Beijing, 100191, China}
   }

 \AIAApapernumber{YEAR-NUMBER}
 \AIAAconference{Conference Name, Date, and Location}
 \AIAAcopyright{\AIAAcopyrightD{YEAR}}

 \newcommand{\eqnref}[1]{(\ref{#1})}

\begin{document}

\maketitle

\newtheorem{thm}{Theorem}
\newtheorem{defn}{Definition} 
\newtheorem{lem}{Lemma}
\newtheorem{assump}{Assumption}
\newtheorem{rem}{Remark}

\begin{abstract}
In this paper, we proposed an alternating projection based algorithm to solve a class of distributed MIN-MAX convex optimization problems. We firstly transform this MIN-MAX problem into the problem of searching for the minimum distance between some hyper-plane and the intersection of the epigraphs of convex functions. The Bregman’s alternating method is employed in our algorithm to achieve the distance by iteratively projecting onto the hyper-plane and the intersection. The projection onto the intersection is obtained by cyclic Dykstra’s projection method. We further apply our algorithm to the minimum time multi-agent consensus problem. The attainable states set for the agent can be transformed into the epigraph of some convex functions, and the search for time-optimal state for consensus satisfies the MIN-MAX problem formulation. Finally, the numerous simulation proves the validity of our algorithm.
\end{abstract}

\section*{Nomenclature}

\begin{tabbing}
  XXXXXX \= \kill
  $\mathbb{R}$ \> Real number field\\
  $\mathbb{R}^{n}$ \> Real coordinate space of n dimensions\\
  $\mathbb{Z_+}$ \> Positive integer set\\
  $f$ \> Convex function \\
  $\mathop{dom}$ \> Domain of a function\\
  $\operatorname{epi}$ \> Epigraph of function\\
  ${{P}}\left( \cdot  \right)$ \> Orthogonal projection onto Chebyshev sets\\
  $\bigcap$ \> Sets intersection operation\\
  $\mathop{dist}\left(\cdot,\cdot\right)$ \> Distance between two sets\\
  $X$ \> Agent state space\\
  $x$ \> Agent state vector \\
  $\dot x$ \> Agent state derivative vector \\
  $u$ \> Control input vector \\
  $\left\| \cdot  \right\|$ \> Euclidean norm \\[5pt]
  \textit{Subscript}\\
  $i$ \> Index number \\
 \end{tabbing}

\section{Introduction}

The distributed convex optimization studies the problem that all agents work cooperatively to minimize some optimal function to agents' local functions. Most existing works concentrate on minimizing the sum of all local function in \citen{nedic2013distributed,lee2013distributed,shi2013reaching,shi2012randomized}. However, MIN-MAX optimization that can be widely applied to portfolio optimization, control system design, engineering design\cite{hare2013derivative} also attract much more attentions as showed in \citen{srivastava2011distributed,hare2013derivative,srivastava2013distributed}. 

In this paper, we propose a cyclic alternating projection based algorithm for distributed convex optimization problem. Alternating algorithm is widely applied to optimal approximation, e.g., solving linear system \cite{tam2012method}, linear programming \cite{tam2012method}, signal processing \cite{combettes2011proximal} and even Sudoku puzzle \cite{tam2012method}. Through the geometric interpretation, we show that the original problem can be reproduced as the problem of searching the minimum distance between the hyper-plane and the intersection of all epigraphs of local functions ${{f}_{i}}$. In our algorithm, we utilize Bregman's alternating projection to obtain the distance. During the procedure, the metric projection onto the intersection is necessary, so we employ another projection algorithm$-$Dykstra's alternating projection as the intermediate procedure. We show that in our algorithm, agents can have the consensus on the point that achieves the minimum distance with only the information from neighbors on a cyclic interaction topology. Thus, the optimal solution of the original problem is also achieved.

Further, we apply our algorithm to the time-optimal consensus problem. In the past decades, a large amount of attention has been devoted to coordinate control of multi-agent systems\cite{cao2013overview,ren2005survey,kranakis2006mobile}. The principal object of coordinated control is to allow the multi-agent work together and coordinate their behaviors in a cooperative fashion to achieve a common goal efficiently. Multi-agent coordination control consists of widespread research fields, including mission assignment, formation control, rendezvous control, consensus and distributed estimation, etc. The consensus problem as the fundamental problem has received greatly development. In contrast to asymptotic, to achieve the consensus with less time cost becomes attractive and reality for practical implementation. Nowadays, numerous works solve the finite time consensus instead. Though the agents can achieve consensus within finite steps through the vanishing speed\cite{li2014distributed} or the Lyapunov function derivative larger than some positive term\cite{xiao2013finite},  the time cost is still not optimal especially with input saturation constraints. To achieve consensus with minimum time cost, \citen{sundaram2007finite} proposed a minimal polynomial based observer to the consensus point to ensure time-optimal. 

Different from the work mentioned above, we directly search the time-optimal consensus state according to the ability of each agent and the agents move towards the state with optimal control. We show that if let the attainable states set, the function to time, as the local functions for the agents, the time-optimal consensus problem is the (see in convex optimization)distributed convex optimization problem. 

This paper is organized as follows. In Section II, the problem formulation is proposed. Section III presents some preliminary on the epigraph and the alternating projection methods. The geometric interpretation on the problem and the MIN-MAX distributed convex optimization algorithm are proposed in Section IV. The application to the time-optimal consensus problem is provided as the proof of algorithm's efficiency in Section V. Finally in Section V, conclusions are provided.

\subsection{Problem Formulation}

In this paper, we consider the distributed multi-agent MIN-MAX convex optimization problems which is presented by

\begin{equation}\label{EQ_OriginalProblem}
\mathop {\min }\limits_{x \in X} \mathop {\max }\limits_{i \in \mathbb{Z_+} } {f_i}\left( x \right),
\end{equation}
where ${{f}_{i}}$ are convex functions that can be only known by agent $i$ and share the same domain, $X \subseteq {R^n}$ is a closed and convex set known by all agents. We assume the problem \eqnref{EQ_OriginalProblem} is well-posed such that $x^*$ achieves the minimum and the feasible solution is finite, i.e. $\left\| {\min \max {f_i}\left( x \right)} \right\| < \infty$. To achieve the overall optimality of the problem, the agents should work cooperatively. The difficulty comes in the distributed formulation that every agent can only access to their own functions and communicate with neighborhood according to the interaction topology. Therefore, agents execute their own local algorithms with only limited knowledge. We make further assumption on the communication topology:
\begin{assump}\label{Assump_1}
The communication between agents is limited in a cyclic digraph interaction topology. Without loss of generality, the interaction sequence is according to the number assignment to the agent, i.e., agent $i$ receiving information from $i-1$ ,$i\in 2,3,\ldots,N$ and agent $1$ receiving from $N$.
\end{assump}

\section{Preliminary}

Before proposing our algorithm, the several background information are presented as followed.

\subsection{Epigraph\cite{boyd2004convex}}
The epigraph of $f:{{R}^{n}}\to R$ is defined as

\begin{equation}\label{EQ_Epgraph1}
\operatorname{epi}f = \left\{ {\left( {x,t} \right)|x \in domf,f\left( x \right) \le t} \right\},
\end{equation}
which is a subset of ${{R}^{n+1}}$. A function is convex if and only if its epigraph is a convex set. Moreover, in terms of epigraphs, the point-wise maximum of convex functions corresponds to the intersection of epigraphs, which is also convex, we have

\begin{equation}\label{EQ_Epgraph2}
\operatorname{epi}\underset{i}{\mathop{\max}}\,{{f}_{i}}\left(x\right)=\underset{i}{\mathop{\bigcap }}\,\operatorname{epi}{{f}_{i}}\left( x \right).
\end{equation}

\subsection{Alternating projection algorithm}
Alternating projection algorithm is a type of geometric optimization
method. Through iteratively orthogonally projecting onto finite number of
Hilbert spaces successively in cyclic setting, the limit to the
projection sequence provides an approximation of the initial point
to those spaces.

Bregman's alternating projection, known as Bregman's algorithm or
Bregman's method designed for closed convex sets is always used to
obtain a point in the intersection of convex sets. In considering
two convex sets without intersection, Bregman's algorithm achieve
the distance between the two sets \cite{boyd2003alternating}.
Following theorem provides detailed descriptions.

Assume there are two convex sets $A,B\subseteq{\mathbb{R}^{n}}$, and
${{P}_{A}}\left( \cdot  \right),{{P}_{B}}\left( \cdot \right)$
denote projection on $A$ and $B$, respectively. We have the
following theorem on above sequences:

\begin{thm}\cite{cao2013overview}\label{Thm_Bregman}
Let $A,B\subseteq {{R}^{n}}$ be closed convex sets and $\left\{
{{a}_{n}} \right\}_{n=1}^{\infty }$,$\left\{ {{b}_{n}}
\right\}_{n=1}^{\infty }$ be the sequences generated by alternating
projection onto $A$ and $B$ from any intimal point ${{x}_{0}}\in
{\mathbb{R}^{n}}$:
\begin{eqnarray}\label{EQ_Bregman1}
{{a}_{n}} &=& {{P}_{A}}\left( {{b}_{n-1}} \right),\nonumber \\
{{b}_{n}} &=& {{P}_{B}}\left( {{a}_{n}} \right),\nonumber\\
{{a}_{1}} &=& {{P}_{A}}\left( {{x}_{0}} \right).
\end{eqnarray}
1. If $A\bigcap B\ne \emptyset$,
\begin{equation} \label{EQ_Bregman2}
{{a}_{n}},{{b}_{n}}\to {{x}^{*}}\in A\bigcap B.
\end{equation}
2. if $A\bigcap B = \emptyset$,
\begin{equation} \label{EQ_Bregman3}
{{a}_{n}}\to {{a}^{*}}\in A,{{b}_{n}}\to {{b}^{*}}\in B,
\end{equation}
where $\left\| {{a}^{*}}-{{b}^{*}} \right\|=dist(A,B)$.
\end{thm}

Ordinary alternating projection can only achieve some point
arbitrarily on the intersection but not the orthogonal projection,
so we employ another variant projection algorithm$-$Dykstra's
alternating projection. This method is usually employed to the
problem

\begin{equation} \label{EQ_DykstraProblem}
{\underset{x\in R}{\text{minimize}}}\,{{\left\| x-r \right\|}^{2}}\ subject\ to\
x\in \bigcap _{i=1}^{n}{{A}_{i}},
\end{equation}
which provides the best approximations to the sets. Recently, Dykstra's algorithm has been extended to solve
least-squares \cite{escalante1996dykstra}, convex optimization
\cite{boyd2011distributed}, etc..

Dykstra's alternating projection implements correction at each
projection to Bregman's method by subtracting the variable, i.e.,
increment. Following theorem provides detailed descriptions.

\begin{thm}\cite{cao2013overview}\label{Thm_Dykstra}
Let ${{A}_{1}},{{A}_{2}},\ldots {{A}_{n}}\subseteq {{R}^{n}}$ be the
closed convex sets with nonempty intersection. Given $x\in {R}^{n}$
iterate by
\begin{eqnarray}\label{EQ_Dykstra1}
x_{n}^{i}&:=&{{P}_{{{A}_{i}}}}\left( x_{n}^{i-1}-I_{n-1}^{i}\right),\nonumber\\
I_{n}^{i}&:=&x_{n}^{i}-\left( x_{n}^{i-1}-I_{n-1}^{i} \right),\nonumber\\
x_{n}^{0}&:=&x_{n-1}^{r},
\end{eqnarray}
with initial values $x_{1}^{0}:=x$, $I_{0}^{i}:=0$ then

\begin{equation}\label{EQ_Dykstra2}
{{x}_{n}}\to {{P}_{\bigcap _{i=1}^{n}{{A}_{i}}}}\left( x \right).
\end{equation}

\end{thm}

\section{MIN-MAX distributed convex optimization algorithm}

In this section, we are going to interpret the original problem \eqnref{EQ_OriginalProblem} geometrically as the distance between some hyper-plane and the intersection of all epigraphs of ${{f}_{i}}$. To solve this geometric problem, we proposed our distributed alternating projection based algorithm. The proof of our algorithm guarantees the optimality of the solution.

\subsection{Geometric interpretation of the MIN-MAX Problem}

From the definition of epigraph in \eqnref{EQ_Epgraph1}, the function $f$ can be regarded as the lower boundary of its epigraph,
\begin{equation}\label{EQ_GeoInterpret1}
\underset{x\in X}{\mathop{\min }} f\left( x \right)=\underset{x\in X}{\mathop{\min }} \left( \operatorname{epi}f\left( x \right)\centerdot \left[ \begin{aligned}
  & {{0}_{n\times 1}} \\
 & 1 \\
\end{aligned} \right] \right).
\end{equation}
Therefore, substitute the point-wise maximum $\underset{i \in \mathbb{Z_+}}{\mathop{\max}}\,{{f}_{i}}\left( x \right)$ into \eqnref{EQ_GeoInterpret1},
\begin{equation}\label{EQ_GeoInterpret2}
\mathop {\min }\limits_{x \in X} {\mkern 1mu} \mathop {\max }\limits_{i \in \mathbb{Z_+}} {\mkern 1mu} {f_i}\left( x \right) = \mathop {\min }\limits_{x \in X} \left( \left[ {{\mathop{\rm epi}\nolimits} \mathop {\max }\limits_{i \in \mathbb{Z_+}} {\mkern 1mu} {f_i}\left( x \right)} \right] \centerdot \left[ \begin{array}{l}
{0_{n \times 1}}\\
1
\end{array} \right] \right).
\end{equation}
Further, since $\underset{i\in \mathbb{Z_+}}{\mathop{\max }}\,{{f}_{i}}\left( x \right)$ is the point-wise maximum of finite number of convex functions ${{f}_{i}}$, apply \eqnref{EQ_Epgraph2} to \eqnref{EQ_GeoInterpret2}

\begin{equation}\label{EQ_GeoInterpret3}
\underset{x\in X}{\mathop{\min }}\,\underset{i\in \mathbb{Z_+}}{\mathop{\max }}\,{{f}_{i}}\left( x \right)=\underset{x\in X}{\mathop{\min}} \left(\bigcap\operatorname{epi}{{f}_{i}}\left( x \right) \centerdot \left[ \begin{aligned}
  & {{0}_{n\times 1}} \\
 & 1 \\
\end{aligned} \right]\right).
\end{equation}

Then the original problem \eqnref{EQ_OriginalProblem} is equivalent to
\begin{equation}\label{EQ_NewProblem}
\underset{x\in X}{\mathop{\min }}\left( \bigcap \operatorname{epi}{{f}_{i}}\left( x \right) \centerdot \left[ \begin{aligned}
  & {{0}_{n\times 1}} \\
 & 1 \\
\end{aligned} \right]\right)
\end{equation}
The geometric problem \eqnref{EQ_NewProblem} can be regarded as the problem of searching the lowest point of the intersection $\bigcap\operatorname{epi}{{f}_{i}}\left( x \right)$.
We notice that the supporting plane to the convex set $\bigcap\operatorname{epi}{{f}_{i}}\left( x \right)$ at the lowest point $\left( {{x}^{*}},{\mathop{\min }}\bigcap\operatorname{epi}{{f}_{i}}\left( x^* \right)\right)$ is
\begin{equation}\label{EQ_GeoSupporting}
\left\{ \left( x,t \right)|t={\mathop{\min }}\left(\bigcap\operatorname{epi}{{f}_{i}}\left( x^* \right)\centerdot \left[ \begin{aligned}
  & {{0}_{n\times 1}} \\
 & 1 \\
\end{aligned} \right]\right) \right\},
\end{equation}
where $x^*$ also achieves the feasible solution to \eqnref{EQ_OriginalProblem}. Therefore, the point that achieves the minimum distance between the convex set $\bigcap\operatorname{epi}{{f}_{i}}\left( x \right)$ to any hyper-plane below and parallel to the supporting plane 
is exactly  $\left( {{x}^{*}},{\mathop{\min }}\bigcap\operatorname{epi}{{f}_{i}}\left( x^* \right)\right)$.

Now, we are ready to express the original problem \eqnref{EQ_OriginalProblem} equivalently as,

\begin{eqnarray}\label{EQ_GeoProblemInterpret}
\mathop{minimize} & \mathop{dist}\left( { \cap epi{f_i}\left( x \right),\left\{ {\left( {x,t} \right)|t = {t_{\min }}} \right\}} \right)\nonumber\\
subject~to & {{t}_{\min }}\le {\mathop{\min }}\left(\bigcap\operatorname{epi}{{f}_{i}}\left( x^* \right)\centerdot \left[ \begin{aligned}
  & {{0}_{n\times 1}} \\
 & 1 \\
\end{aligned} \right]\right).
\end{eqnarray}
If we can find the point that achieves the minimum distance, the feasible solution is also obtained.

\begin{rem}
If the the solution to \eqnref{EQ_NewProblem} is greater than zero, it is exactly the minimum distance between the epigraphs intersection $\bigcap\operatorname{epi}{{f}_{i}}\left( x \right)$ and the zero-time plane $\left\{ {\left( {x,t} \right)|t = 0} \right\}$.
\end{rem}

\subsection{MIN-MAX distributed convex optimization algorithm}

Recall the discussions in section II that Bregman’s method can obtain the distance of two convex sets. Since both the intersection and the hyper-plane are convex sets, we can directly apply Bregman's method to \eqnref{EQ_GeoProblemInterpret}. However, the procedure in that method requires the projection onto the intersection $\bigcap\operatorname{epi}{{f}_{i}}\left( x \right)$, which is difficult to obtain immediately especially in the distributed setting. To achieve this goal, our distributed algorithm brings in cyclic Dykstra’s projection method as the intermediate part in Bregman’s method which iteratively projects onto the intersection and the hyper-plane. According to Theorem \ref{Thm_Bregman} and \ref{Thm_Dykstra}, the algorithm for each agent performance under Assumption \ref{Assump_1} is proposed in details with following three steps:

Initialization: The hyper-plane is firstly determined by choosing a real number $t_{\min}$ small enough. Agent $i$ maintains its own guess on the point $\left(x_i,t_i\right)$ that achieves the minimum of the distance in \eqnref{EQ_GeoProblemInterpret} and the increment $I_{n}^{i}$ mentioned in \eqnref{EQ_Dykstra1}. The flag $flag$ called increment reset symbol that indicates whether the increment preserves or resets to zero passes over the network. Without loss of generality, Agent $1$ remembers its previous guess, and after comparing with the current guess it can decide the algorithm stopping time. The communication topology is cyclic and the order is determined by the indexes assigned to the agents.

1.	Agent $i$ receives the guess of the point ${\left( {x,t} \right)_n^{i - 1} - I_{n - 1}^i}$ and the increment reset symbol from previous agent $i-1$. Agent $i$ preforms the procedure described in Theorem 2, i.e., projecting ${\left( {x,t} \right)_n^{i - 1} - I_{n - 1}^i}$ with the increment onto its own function epigraph $\operatorname{epi}{{f}_{i}}\left( x \right)$ to get its guess $\left( {x,t} \right)_n^i$, and updating the increment whether “reset” or “preserve” according to the symbol, 
\begin{eqnarray}\label{Eq_Algorithm1}
\left( {x,t} \right)_n^i = {P_{{A_i}}}\left( {\left( {x,t} \right)_n^{i - 1} - I_{n - 1}^i} \right),\\
I_n^i = \left\{ 
\begin{aligned}\label{Eq_Algorithm2}
&\left( {x,t} \right)_n^i - \left( {\left( {x,t} \right)_n^{i - 1} - I_{n - 1}^i} \right) & ~flag = 0\\
&0 & ~flag = 1
\end{aligned} \right..
\end{eqnarray}

2.	According to the cyclic interaction topology,agent $i$ passes the projection $\left( {x,t} \right)_n^i$ and the increment reset symbol $flag$ to the next agent $i+1$.

3.	During agent $1$'s turn in each iteration cycle, agent $1$ compares its new updated projection with its previous guess. If the error between them is acceptable, agent $1$ resets the increment to 1(“reset”), and projects the new updated projection onto the hyper-plane $\left\{ {\left( {x,t} \right)|t = {t_{\min }}} \right\}$, otherwise Agent $1$ receives guess from Agent $N$ and repeats step 1,
\begin{eqnarray}\label{Eq_Algorithm3}
  & e=x_{n}^{1}-x_{n-1}^{1} \\
 & flag=\left\{ 
\begin{aligned}\label{Eq_Algorithm4}
  & 1~~&\left\| e \right\|<err \\
 & 0~~&\left\| e \right\|\ge err \\
\end{aligned} \right.
\end{eqnarray}

The first two steps are the implementation of Dykstra’s algorithm. In step 3, once the projection of the intersection of all convex functions’ epigraph is obtained that is to say the error between each iteration is small enough, we implement Bregman’s algorithm by projecting the result of Dykstra’s algorithm onto the hyper-plane. Taken the projection on the hyper-plane as the new guess of the feasible solution, the agents iteratively perform step 1-3, and finally obtained the feasible solution to problem \eqref{EQ_GeoProblemInterpret}.

The correctness of algorithm to problem can be provided by following theorem:

\begin{thm}
If the functions ${{f}_{i}}\left( x \right)$ are convex and the hyper-plane satisfies \eqnref{EQ_GeoProblemInterpret}, the convex set $\bigcap\operatorname{epi}{{f}_{i}}\left( x \right)$ has non-intersection with the hyper-plane. Give $x$ iterate by

\begin{eqnarray}
  & \left( {{x}_{p}},{{t}_{p}} \right)={{P}_{\bigcap epi{{f}_{i}}}}\left( {{x}_{n-1}},{{t}_{n-1}} \right), \label{Eq_AlgorithmProof1}\\
 & \left( {{x}_{n}},{{t}_{n}} \right)={{P}_{s}}\left( {{x}_{n}},{{t}_{n}} \right)=\left( {{x}_{p}},t \right), \label{Eq_AlgorithmProof2}
\end{eqnarray}
where ${{P}_{\bigcap epi{{f}_{i}}}}\left( {\centerdot} \right)$ is obtained by the procedure in theorem 2. Then
\begin{equation}\label{Eq_AlgorithmProof3}
{{x}_{n}}\to \underset{x}{\mathop{\arg \min }}\,\max {{f}_{i}}.
\end{equation}
\end{thm}
\begin{proof}
Since $\bigcap\operatorname{epi}{{f}_{i}}\left( x \right)$ and the hyper-plane are convex sets, it can be immediately derived from the definition that they are separated. Apparently, \eqnref{Eq_AlgorithmProof1} and \eqnref{Eq_AlgorithmProof2} the implementation of theorem 1, and the projection in \eqnref{Eq_AlgorithmProof2} is obtained from Theorem 2. Consequently, $\left( {{x_n},{t_n}} \right)$ approaches the point that achieves the minimum distance between the intersection and the hyper-plane. As mentioned in previous subsection, the first part of that point is the feasible solution to the problem \eqref{EQ_OriginalProblem}, therefore the limit of $x_n$ satisfies \eqref{Eq_AlgorithmProof3}.
\end{proof}

\section{Application to the minimum time consensus problem}
The multi-agent consensus problem has attracted great attentions during the last decade, and achieved remarkable development and success. The consensus problem concentrates on the distributed negotiation between agents that can finally achieve consensus on some coordinated variables. With this fundamental concept, numerous problems can be treated by consensus such as formation, rendezvous, alignment problems, etc. However, the negotiation result and time consumption is unpredicted before the process, and influenced by communication topology and their initial state of coordinated variables. With this character, we are interested in how to achieve the minimum time consensus with admissible control input, namely time-optimal consensus problem. Actually, if without the bounded control input constraint, the consensus can be achieved within any limited time. We firstly present this problem as an MIN-MAX distributed optimization problem, and show the consensus problems with first order agents or second order agents with zero velocity constraints to the initial and final state are can be treated as Problem 1 which can be solved by our algorithm. After that, demonstrations are provided to test the validity of our algorithm.

Let the minimum reachable time be the function ${{f}_{i}}$ of the states in the state-space for agent $i$, and then the minimum reachable time for all agents is the maximum of all ${{f}_{i}}$, i.e.,
\begin{equation}
\underset{i}{\mathop{\max }}\,{{f}_{i}}\left( x \right).
\end{equation}
The minimum time for all agents consensus becomes
\begin{equation}
\underset{x}{\mathop{\operatorname{min}}}\,\underset{i}{\mathop{\max }}\,{{f}_{i}}\left( x \right),
\end{equation}
where the state achieves above minimum is the corresponding time-optimal consensus state. If the functions are convex, the minimum time consensus problem is exactly Problem \eqnref{EQ_OriginalProblem}. The attainable set for any linear model agents with admissible control input is a continues function $\Omega\left(t\right)$, and for specific time instance $t$ the set is closed, bounded and strictly convex.\cite{meschler1963time} Apparently, the attainable set is the epigraph of ${{f}_{i}}\left(t\right)$, therefore the time-optimal consensus problem can be also interpreted according to \eqnref{EQ_GeoProblemInterpret} as the minimizing the distance between the intersection of attainable sets and the hyper-plane.
The following part will analyse the first-order and second-order systems separately.

\subsection{First order systems}
Consider following first order system model

\begin{equation}
\dot{x}=u, ~~ -{{u}_{\max }}\le u\le {{u}_{\max }}.
\end{equation}

Agent $i$ starts from the initial state ${{x}_{0}}$ to ${{x}_{1}}$ with the minimum reach-time
\begin{equation}
\frac{{{\left\| {{x}_{0}}-{{x}_{1}} \right\|}_{2}}}{{{u}_{\max }}}.
\end{equation}
Its attainable state set is provided as,
\begin{equation}
\left\{ {\left( {x,t} \right)|\left\| {\frac{{x - {x_0}}}{{{u_{\max }}}} \le t} \right\|,t \ge 0} \right\},
\end{equation}
which forms a cone in state-time space. Since cones are convex sets, we can directly implement our algorithms (\ref{Eq_Algorithm1}-\ref{Eq_Algorithm4}). From any initial state, the agents apply the algorithm to calculate the time-optimal state for consensus, and form their own optimal control, normally saturation control, towards the consensus state.
Since the result is quite simple such that it can be covered by following subsection, and due to the page limitation, the simulation for first order agents is eliminated.
\subsection{Second-order Systems}
Consider the second order system model

\begin{equation}
\left\{ \begin{array}{l}
{{\dot x}_1} = {x_2}\\
{{\dot x}_2} = u
\end{array} \right., - {u_{\max }} \le u \le {u_{\max }}
\end{equation}

Agent $i$ starts from an arbitrary initial states $\left( {{x}_{11}},{{x}_{21}} \right)$ to the arbitrary states $\left( {{x}_{12}},{{x}_{22}} \right)$. The time optimal control should be bang-bang strategy which is to execute a max positive (negative) input for a period of time and then to reverse the input for another period of time.

The control law is provided by following equation:

\begin{equation}
u=\mathop{sgn} \left( {{x}_{12}}-x_1-\frac{1}{2}\left( x_2+{{x}_{22}} \right)\left\| x_2-{{x}_{22}} \right\| \right){{u}_{\max }}
\end{equation}

The optimal time-cost is determined by the relation between initial and final state.
\begin{equation}\label{Eq_SecOrdOptTime}
\left\{ \begin{aligned}
{\left[ {t + \left( {{x_{22}} + {x_{21}}} \right)} \right]^2} = 4\left( {{x_{12}} - {x_{11}}} \right) + 2x_{21}^2 + 2x_{22}^2,~~ & {x_{12}} - {x_{11}} - \frac{1}{2}\left( {{x_{21}} + {x_{22}}} \right)\left\| {{x_{21}} - {x_{22}}} \right\| < 0\\
{\left[ {t - \left( {{x_{21}} + {x_{22}}} \right)} \right]^2} =  - 4\left( {{x_{12}} - {x_{11}}} \right) + 2x_{21}^2 + 2x_{22}^2,~~ & {x_{12}} - {x_{11}} - \frac{1}{2}\left( {{x_{21}} + {x_{22}}} \right)\left\| {{x_{21}} - {x_{22}}} \right\| > 0
\end{aligned} \right.
\end{equation}
%
%
%

The attainable states set is presented as,
\begin{eqnarray}\label{Eq_SecOrdEpigraph}
\{\left({x,t}\right)|{\left[ {t + \left( {{x_{22}} + {x_{21}}} \right)} \right]^2} \ge 4\left( {{x_{12}} - {x_{11}}} \right) + 2x_{21}^2 + 2x_{22}^2,&{x_{12}} - {x_{11}} - \frac{1}{2}\left( {{x_{21}} + {x_{22}}} \right)\left\| {{x_{21}} - {x_{22}}} \right\| < 0\nonumber\\
{\left[{t - \left( {{x_{21}} + {x_{22}}} \right)} \right]^2} \ge  - 4\left( {{x_{12}} - {x_{11}}} \right) + 2x_{21}^2 + 2x_{22}^2,&{x_{12}} - {x_{11}} - \frac{1}{2}\left( {{x_{21}} + {x_{22}}} \right)\left\| {{x_{21}} - {x_{22}}} \right\| > 0\}
\end{eqnarray}

Unfortunately, the epigraph or attainable set \eqnref{Eq_SecOrdEpigraph} is not convex. Even when the problem has been simplified such that the agents have zero initial and final velocities, i.e. ${{x}_{21}}={{x}_{22}}=0$, the epigraph composed by two mirror symmetry half-parabola,
\begin{equation}\label{Eq_SecOrdEasyEpigraph}
{{t}^{2}}\ge4\left\| {{x}_{12}}-{{x}_{11}} \right\|
\end{equation}
is sill not convex, and we could not directly implement our algorithm. But we notice that if we apply following quadratic transformation to the time
\begin{equation}\label{Eq_SecOrdEasyTrans}
s={{t}^{2}},
\end{equation}
the attainable state set becomes a convex function to time square.
\begin{figure}
 \centering
  \includegraphics[width=9cm,height=5cm]{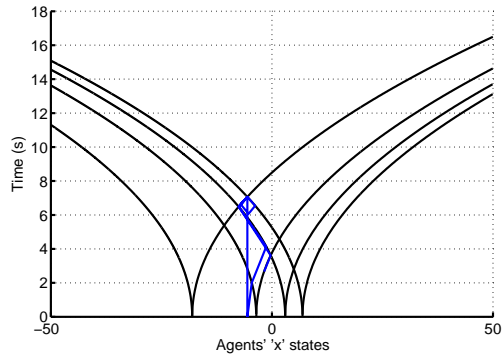}
 \caption{Projecting Result of the Algorithm onto Attainable sets}
 \label{f:Zero-ProjectionResult}
\end{figure}

We assume that each agent with initial state $x_i$, and their attainable sets $s_i$ is presented as
\begin{equation}\label{Eq_SecOrdEasyNewEpi}
s_i\ge 4\left\| {x}-{{x}_{i}} \right\|,
\end{equation}
where $s_i$ is the square of real time $t$. Therefore, we can substitute the new epigraph \eqnref{Eq_SecOrdEasyNewEpi} into \eqnref{EQ_GeoProblemInterpret}. Further, since the consensus time larger than zero, the hyper-plane can be chosen as the zero time plane $\left\{\left(x,t\right)| t = 0\right\}$.

Consider four agents, whose initial states are
${x_1} = \left( { - 3.542884,0} \right)$, 
${x_2} = \left( {3.001152,0} \right)$, 
${x_3} = \left( {6.924106,0} \right)$, 
${x_4} = \left( { - 18.0296,0} \right)$, 
respectively. The time-optimal consensus state is $\left(-5.5527,0\right)$ and the optimal time is
7.0645s.

We apply our algorithm to those four agents, and the demonstration result is presented as followed.

The projections on the intersection and the hyper-plane is presented in Fig. \ref{f:Zero-ProjectionResult}. Through the transformation \eqnref{Eq_SecOrdEasyTrans}, the problem can be treated as the convex functions and finally achieves the minimum of intersection. We could find that the Bregman's method is only proceeded 2 times, which is quite efficient.
\begin{figure}
 \centering
  \includegraphics[width=9cm,height=5cm]{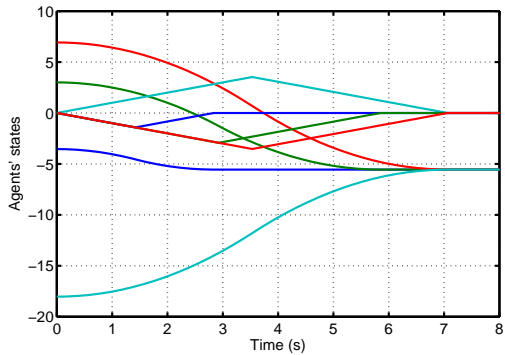}
 \caption{State-Trajectory of the Agents with Optimal Control}
 \label{f:zero_state}
\end{figure}

\begin{figure}
 \centering
  \includegraphics[width=9cm,height=5cm]{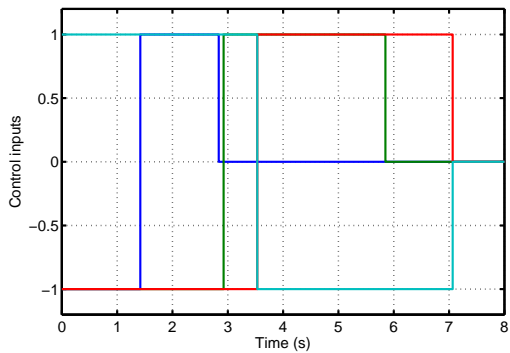}
 \caption{Control Sequences of the Agents}
 \label{f:zero_input}
\end{figure}
The state-trajectory and control sequences are presented in Fig. \ref{f:zero_state} and \ref{f:zero_input}. The control inputs are typically bang-bang control, and the agents achieves consensus with the minimum time.

%

To make a further step, we can implement our algorithm without proof to the more complicated non-convex case such that only the final velocities are assumed to be zero. When we apply different transformations, though this case does not satisfy the convex assumption, the result is quite inspiring.

Let ${{x}_{22}}=0$ in equations \eqnref{Eq_SecOrdOptTime} and \eqnref{Eq_SecOrdEpigraph}. The optimal time-cost curves and the corresponding epigraphs are 

\begin{equation}\label{Eq_SecOrdHardOptTime}
\left\{ \begin{aligned}
{\left[ {t +  x_{22}} \right]^2} = 4\left( {{x_{12}} - {x_{11}}} \right) + 2x_{21}^2 ,~~ & {x_{12}} - {x_{11}} - \frac{1}{2}x_{21}\left\|x_{21}\right\| < 0\\
{\left[ {t - x_{21})} \right]^2} =-4\left( {{x_{12}} - {x_{11}}} \right) + 2x_{21}^2 ,~~ & {x_{12}} - {x_{11}} - \frac{1}{2}x_{21}\left\|x_{21}\right\| > 0
\end{aligned} \right.,
\end{equation}
and
\begin{eqnarray}\label{Eq_SecOrdHardEpigraph}
\{\left({x,t}\right)|{\left[ {t + x_{21}} \right]^2} \ge 4\left( {{x_{12}} - {x_{11}}} \right) + 2x_{21}^2 ,&{x_{12}} - {x_{11}} - \frac{1}{2} x_{21}\left\| {{x_{21}}} \right\| < 0\nonumber\\
{\left[{t - x_{21}} \right]^2} \ge  - 4\left( {{x_{12}} - {x_{11}}} \right) + 2x_{21}^2 ,&{x_{12}} - {x_{11}} - \frac{1}{2} x_{21}\left\| {{x_{21}}} \right\| > 0\}.
\end{eqnarray}

%
%
%

Because the half-parabolas in \eqref{Eq_SecOrdHardEpigraph} are with different vertexes, it’s hard to find a simple transformation like \eqref{Eq_SecOrdEasyTrans} to make the epigraph to be cones in this case. For the consideration of convex assumption, we apply different quadratic transformations to those two parts, left and right, of different agents' epigraphs.
if ${x_{i{1_{normal}}}} > {x_{i1}} + \frac{1}{2}{x_{i2}}\left\| {{x_{i2}}} \right\|$ 
\begin{equation}
{x_{i{1_{con}}}} = {x_{i{1_{normal}}}}\\
{x_{i{2_{con}}}} = \left\{ 
\begin{aligned}
\frac{1}{4}{\left( {{x_{i{2_{normal}}}} + {x_{i2}}} \right)^2} + \left\| {{x_{i2}}} \right\| - \frac{1}{2}\left( {x_{i2}^2 - {x_{i2}}\left\| {{x_{i2}}} \right\|} \right) & {x_{i{2_{normal}}}} \ge  - {x_{i2}}\\
\left\| {{x_{i2}}} \right\| - \frac{1}{2}\left( {x_{i2}^2 - {x_{i2}}\left\| {{x_{i2}}} \right\|} \right) & {x_{i{2_{normal}}}} <  - {x_{i2}}
\end{aligned} \right.
\end{equation}
where $x_i=\left(x_{i1},x_{i2}\right)$ is the initial state of agent $i$, $x_{i_{normal}}=\left(x_{i1_{normal}},x_{i1_{normal}}\right)$ is the projection in normal axis and $x_{i_{con}}=\left(x_{i1_{con}},x_{i1_{con}}\right)$ is the in transformed axis.

\begin{figure}
 \centering
  \includegraphics[width=8cm,height=8.67cm]{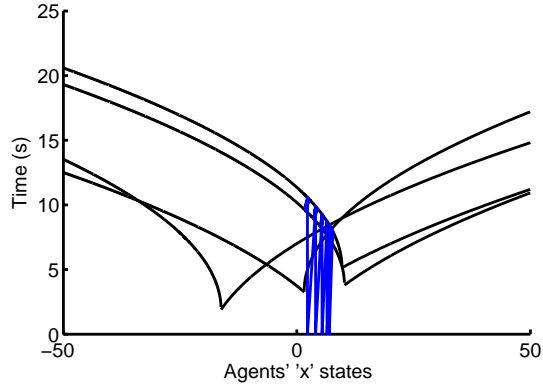}
 \caption{Projecting Result of the Algorithm onto Attainable sets}
 \label{f:ProjectionResult}
\end{figure}

Consider four agents, whose initial states are
${x_1} = \left( { - 3.542884,5.140490} \right)$,
${x_2} = \left( {3.001152,3.794066} \right)$,
${x_3} = \left( {6.924106,-3.281824} \right)$,
${x_4} = \left( { - 18.0296,1.9023} \right)$,
respectively. The time-optimal consensus state is $\left(6.9366,0\right)$ and the optimal time is
8.4467s.

We apply our algorithm to those four agents, and the demonstration result is presented as followed.
The projections on the intersection and the hyper-plane is presented in Fig. \ref{f:ProjectionResult}. Though we applied different transformation to the epigraph of the agent, and the transformation was distinguish for different agents, the projection still converged to the minimum of the epigraph intersection. We believe the ability of non-convex treatment is releated to the characteristics of the alternating projection.
\begin{figure}
 \centering
  \includegraphics[width=8cm,height=8.67cm]{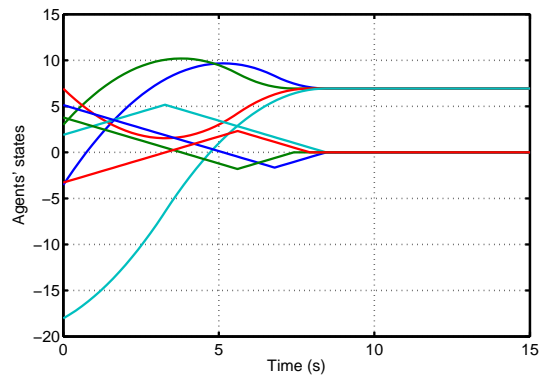}
 \caption{State Trajectory of the Agents with Optimal Control}
 \label{f:arbi_state}
\end{figure}

\begin{figure}
 \centering
  \includegraphics[width=8cm,height=8.67cm]{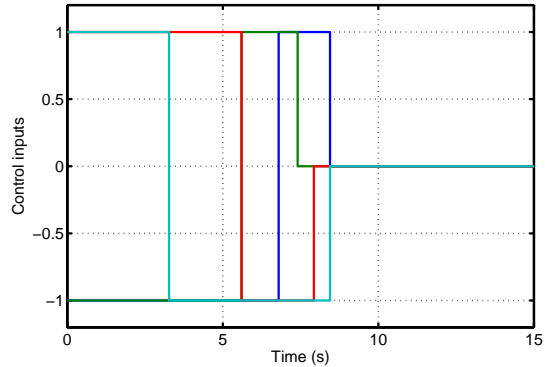}
 \caption{Control History of the Agents}
 \label{f:arbi_input}
\end{figure}
The state-trajectory and control sequences are presented in Fig. \ref{f:arbi_state} and \ref{f:arbi_input}, which performed the same as previous example.

\section{Conclusion}
In this paper, we interpreted the distributed MIN-MAX convex optimal problems geometrically as the problem of searching the minimum distance between some hyper-plane and the epigraphs intersection of convex functions. To solve this problem, we proposed an alternating projection based algorithm which is constructed by the Bregman’s and the cyclic Dykstra’s alternating projection method. This distributed algorithm guarantees the optimal solution to the problem with only neighbor-communication on a cyclic communication topology. Moreover, we implement our algorithm to the multi-agent minimum-time consensus problem. We have shown that the first-order systems and the second-order systems with zero initial and final velocities can be formulated into a standard distributed MIN-MAX convex optimal problems. With the finite time attainable state set as the optimize functions, the time optimal state for consensus can be obtained by our algorithm, and each agent can execute optimal control to that state. Heuristically, we also implement our algorithm to the non-convex cases such that the second-order systems achieve consensus on zero velocities state without proof. At last, the demonstrations are presented to illustrate the efficiency and validity of our algorithm, even for the non-convex cases.

%

\section{Acknowledgments}
This paper is based upon work supported by the National Natural
Science Foundation of China (Grant No. 61273349, 61175109).


\bibliography{bibtex_database}
\bibliographystyle{aiaa}

\end{document}